\documentclass[a4paper,DIV=classic]{myart}
\usepackage[dvistyle]{todonotes}

\newcommand{\norm}[1]{\left\Vert#1\right\Vert}
\newcommand{\abs}[1]{\left\vert#1\right\vert}
\newcommand{\Set}[1]{\ensuremath{ \left\{ #1 \right\} }}
\newcommand{\set}[1]{\ensuremath{ \{ #1 \} }}
\newcommand{\R}{\mathbb{R}}
\newcommand{\N}{\mathbb{N}}

\DeclareMathOperator*{\esssup}{ess\,sup}
\DeclareMathOperator*{\essinf}{ess\,inf}

\def\e{\mathrm{e}}
\def\ud{d}
\def\dx{dx}

\def\ds{ds}

\def\du{du}
\newcommand{\C}{\mathbb{C}}

\begin{document}
\title{A Fourier Approach to the Computation of CV@R and Optimized Certainty Equivalents}
\author[a,1,s]{Samuel Drapeau}
\author[b,2]{Michael Kupper}
\author[a,3,t]{Antonis Papapantoleon}

\address[a]{Institute of Mathematics, TU Berlin, Stra\ss e des 17. Juni 136, 10623 Berlin, Germany}
\address[b]{University of Konstanz, Universit\"atstra\ss e 10, 78464 Konstanz} 

\eMail[1]{drapeau@math.tu-berlin.de}
\eMail[2]{kupper@uni-konstanz.de}
\eMail[3]{papapan@math.tu-berlin.de}

\myThanks[s]{Financial support: MATHEON project E.11}
\myThanks[t]{Financial support: MATHEON project E.13}



\abstract{
We consider the class of risk measures associated with optimized certainty equivalents.
This class includes several popular examples, such as CV@R and monotone mean-variance.
Numerical schemes are developed for the computation of these risk measures using Fourier transform methods.
This leads, in particular, to a very competitive method for the calculation of CV@R which is comparable in computational time to the calculation of V@R.
We also develop methods for the efficient computation of risk contributions.
}
\keyWords{V@R, CV@R, Optimized Certainty Equivalent, Fourier Methods, Risk Contribution.}

\maketitle
\section*{Introduction}
\addcontentsline{toc}{section}{Introduction}
\markboth{\uppercase{Introduction}}{\uppercase{Introduction}}

The quantification of risk is more than ever a central issue in modern asset and risk management.
The increasing volume and complexity of financial instruments have raised the need not only for \emph{coherent} but also for \emph{efficient} and \emph{accurate} risk measurement methods.
In the banking industry,  a vast amount of positions and portfolios have to be assessed daily, which makes the computational speed of risk measurement methods a matter of paramount importance.
Starting with Value at Risk (V@R), the goal of risk measures was to quantify the minimal amount of capital required in order to recover from unexpected large losses.
V@R became very popular---see also the Basel II capital requirements---and is nowadays a standard instrument in the industry mainly for two reasons.
Firstly, it has an apparently obvious financial interpretation: it is the minimal amount of capital that has to be added to a position in order to push the probability of losses below a threshold level.
Secondly, it has an easy and fast implementation: given a portfolio distribution, it simply amounts to the computation of the quantile of this distribution at the threshold level.
However, V@R has a very serious deficiency; namely, it does not fulfill the basic property of diversification.
Indeed, it may well happen that V@R delivers a lower risk for a portfolio concentrated in a single asset rather than for one diversified into several assets.

In order to overcome this drawback, \citet*{artzner01} introduced an axiomatic approach to coherent risk measures inciting diversification.
An important example of such a risk measure is the Conditional Value at Risk (CV@R), which is strongly related to the Average Value at Risk and the Expected Shortfall.
The seminal paper on coherent risk measures \citep{artzner01} was later generalized to monetary convex risk measures by \citet*{foellmer02} and \citet*{fritelli03} providing new examples, the most prominent of which are the entropic and the shortfall risk measures.
A notable subclass are spectral or law invariant monetary risk measures, which have additional properties that make them particularly attractive for numerical implementations; see e.g. \citet{kusuoka2001, acerbi2002} and \citet*{Jouini_etal_2006}.
An important application of these new risk measurement methods is the portfolio optimization scheme with respect to CV@R developed by \citet{rockafellar00}.

The literature on numerical methods for risk measures however has mostly concentrated on V@R; see \citet{Glasserman2004} for an overview. 
In the area of credit risk, there is more intense activity on computational methods for CV@R and other coherent or convex risk measures, see e.g. \citet{Kalkbrener}.
Moreover, most of this literature concentrates on simulation-based methods, see 
e.g. \citet*{Bardou_Frikha_Pages_2009} and \citet*{Dunkel_Weber_2010} as well as 
the references therein.\footnote{The problem of statistical robustness for law 
invariant risk measures has been raised by \citet{cont2010}. They showed that 
CV@R, in contrast to V@R, is not continuous with respect to the 
weak${}^\ast$-topology. However this topology is very weak, and recently 
\citet{shied2012} showed that a large class of law invariant risk measures is 
statistically robust for a reasonably strong topology. This is also the case for 
the Optimized Certainty Equivalents used here.}
Compared to V@R though, coherent and convex risk measures are typically more difficult to calculate and more costly in terms of computational time.
Taking CV@R as an example, instead of computing the quantile of the distribution at one level, it accounts for an integration of the quantile function over an interval, which increases significantly the computational complexity.

The goal of this paper is to focus on a specific class of law invariant risk 
measures, the \emph{optimized certainty equivalents} which were introduced by 
\citet{ben-tal01,ben-tal02}, and to use \emph{Fourier transform methods} and 
\textit{deterministic root-finding schemes} in order to compute them 
efficiently.
The first reason for choosing this class is that it contains most of the classical examples: CV@R, the entropic risk measure, and monotone Mean-Variance among others.
The second reason is that, due to its nice smoothness properties, it provides a fairly easy scheme for the numerical computation. 
This can be summarized in the following two steps:
\begin{enumerate}
	\item Solve an allocation problem using a one dimensional root finding algorithm and transform methods;
	\item Based on this optimal allocation, compute an expectation using transform methods.
\end{enumerate}

The terminus `transform method' indicates any method that uses the characteristic or moment generating function of a random variable for the computation of expectations.
This includes the Fourier transform method of \citet*{CarrMadan99}, the Laplace transform method of \citet*{Raible00} and the cosine series expansion of \citet*{Fang_Oosterlee_2009}. 
We will use Fourier transfroms and follow the work of \citet*{EberleinGlauPapapantoleon08} closely, while we refer to \citet*{Schmelzle_2010} for a comprehensive overview and numerous references.
Similarly, the term `root finding algorithm' refers to any method for determining the root of a function; e.g. bisection, secant or Newton's method, cf. \citet*{Stoer_Bulirsch_2002} for an overview.
We will actually use Brent's algorithm, which combines the bisection, the secant and the inverse interpolation methods (see \citet*{Brent_1973}) for determining the roots of equations. 

Transform methods have been introduced to mathematical finance for option pricing, see for instance \citep{CarrMadan99,Raible00}, and have proved a very efficient tool when the moment generating function of the underlying random variable is known.
This is the case, in particular, for infinitely divisible distributions (i.e. L\'evy models) and affine processes.
In the context of risk measurement, the application of transform methods has been largely unexplored; see \citet*{Rachevetal} for an application to the computation of CV@R.
Fourier transform methods turn out to be a very efficient tool for the computation of optimized certainty equivalents as well.
In particular, the calculation of CV@R using Fourier methods has similar computational complexity to the computation of V@R, thus both risk measures can be computed in almost the same amount of time.
This should be a further argument supporting the use of CV@R in practical applications.

This paper is organized as follows: in section \ref{sec:setting} we present the optimized certainty equivalents and their connections to risk measures.
The representation of risk contributions in this framework is also discussed.
In section \ref{sec:fourier}, we develop computational methods for optimized certainty equivalents using Fourier methods and deterministic root-finding algorithms.
We concentrate on the study of the entropic risk measure, conditional value at risk and polynomial risk measures.
We also illustrate the scope of applications by presenting some realistic scenarios where this method applies particularly well, and provide examples for the computation of risk contributions.
In the last section, we compare the computational efficiency and accuracy of the developed schemes with respect to several other methods.

\section{Optimized Certainty Equivalent}
\label{sec:setting}

This section is devoted to the class of risk measures associated with optimized 
certainty equivalents. 
They are induced by a (parametric) loss function which reflects the relative 
risk aversion of an agent.
Optimized certainty equivalents generate naturally quasi-convex risk measures. 
Moreover, we also discuss risk contributions in this framework.

Let $(\Omega,\mathscr{F},P)$ be an atomless probability space.
By $L^0$ we denote the set of random variables identified when they coincide $P$-almost surely.
By $L^p$ we denote the set of those random variables in $L^0$ with finite $p$-norm.

\begin{definition}
	A function $l:\R \to \R$ is called a \emph{loss function} if
	\begin{enumerate}[label=(\roman*)]
	    \item $l$ is increasing and convex;
	    \item $l(0)=0$ and $l (x)\geq x$.
	\end{enumerate}
\end{definition}
The loss function is used to measure the expected loss $E[l(-X)]$ of a financial position $X$.
Therefore, relative to the risk neutral evaluation of losses $E[-X]$, the loss function puts more weight on high losses and less on gains; see Figure \ref{fig:loss_functions}.
This is what the second property of $l$ essentially conveys.
As for the first property, it translates the two normative facts related to 
risk, that `\textit{diversification should not increase risk}' and 
`\textit{the better for sure, the less risky}'.
We refer to \citet{drapeau01} for a discussion about these facts.
In our examples, the loss function will additionally satisfy the following, slightly stronger, assumption 
\begin{enumerate}[label=\textbf{(A)}]
    \item\label{assumption} $l(x)> x$ for all $x$ such that $\abs{x}$ is large enough.
\end{enumerate}
This means that we are strictly more averse than the risk neutral evaluation.
\begin{figure}[ht!]
	\centering
	\includegraphics[width=12cm]{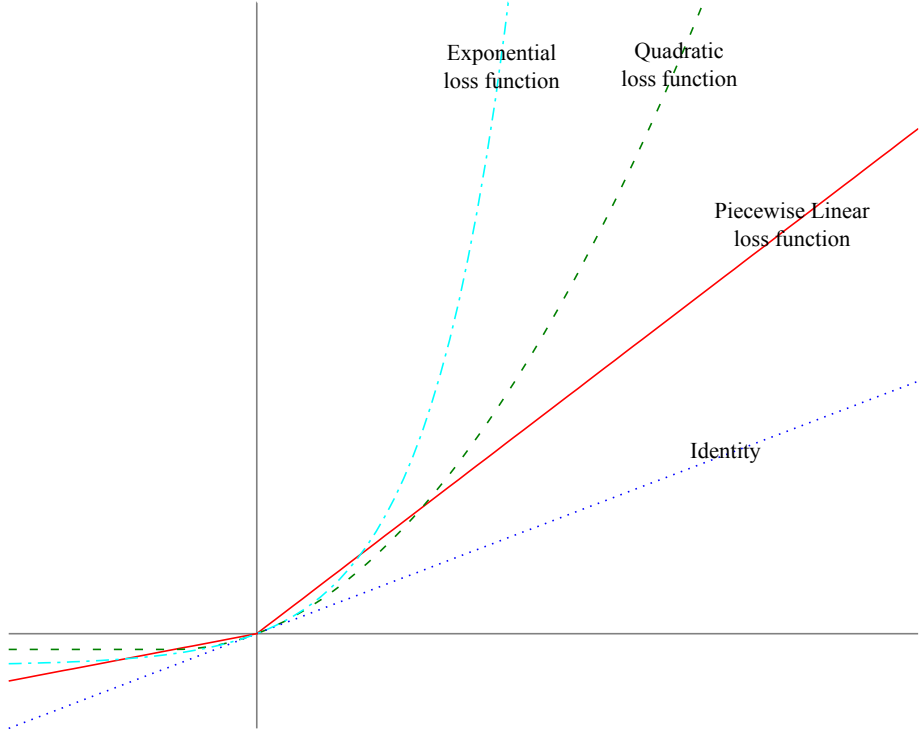}
	\caption{Plot of exponential, quadratic and piecewise linear loss functions (cf. section \ref{sec:fourier}).}
	\label{fig:loss_functions}
\end{figure}

Throughout this paper we will work in the setting of Orlicz spaces which is particularly suitable for optimization in finance and economics, see \citet{kreps01}, \citet{biagini2008} and \citet{cheridito02}.
There are two reasons that motivate this choice in comparison to $L^\infty$.
Firstly, this is the natural setting on which the optimized certainty equivalent is defined and also fits well with Fourier transforms.
Secondly and most importantly, it allows to consider unbounded payoffs which are the rule, rather than the exception, in financial markets.

Let us denote by $l^\ast$ the convex conjugate of $l$, that is,  $l^\ast(y)=\sup_{x\in \R} \set{xy-l(x)}$.
Following \citet{cheridito02,cheridito2008}, we define the Orlicz heart
\begin{equation}
	\mathcal{X}_l:=\Set{X \in L^0: E\left[ l(c\abs{X}) \right] < +\infty\text{ for all } c>0}
	\label{}
\end{equation}
which is, for the $P$-almost sure ordering and the $l$-Luxembourg norm
\begin{equation}
	\norm{X}_{l}:=\inf \Set{ a> 0 : E\left[ l\left( \frac{\abs{X}}{a} \right) \right]\leq 1},
	\label{}
\end{equation}
a Banach lattice.
The norm dual of $\mathcal{X}_l$ is the Orlicz space
\begin{equation}
	\mathcal{X}_{l}^\ast:=\Set{Y \in L^0: E\left[ l^\ast(c\abs{Y}) \right]<+\infty\text{ for some } c>0}
	\label{}
\end{equation}
with the Orlicz norm
\begin{equation}
	\norm{Y}^\ast_{l}:=\sup\Set{E[YX]:\norm{X}_{l}\leq 1},
	\label{}
\end{equation}
which is equivalent to the Luxembourg norm $\norm{\cdot}_{l^\ast}$.
Since $l(x)\geq x$ for all $x \in \R^+$, it follows that $\mathcal{X}_l\subseteq L^1$.

We denote with $\mathcal{M}_{1,l^\ast}(P)$ the set of those probability measures on $\mathscr{F}$ which are absolutely continuous with respect to $P$ and whose densities are in $\mathcal{X}_{l}^\ast$.
We consider risk measures in the following sense.
\begin{definition}
	A \emph{risk measure} is a function $\rho:\mathcal{X}_{l}\to [-\infty,+\infty]$ which is
	\begin{enumerate}[label=(\roman*)]
		\item quasi-convex: $\rho(\lambda X+(1-\lambda)Y)\leq \max\set{\rho(X),\rho(Y)}$ for all $X,Y \in \mathcal{X}_{l}$ and $\lambda \in ]0,1[$;
		\item monotone: $\rho(X)\geq \rho(Y)$ whenever $X\leq Y$ for $X,Y \in \mathcal{X}_{l}$.
	\end{enumerate}
	A risk measure is called \emph{monetary} if it is
	\begin{enumerate}[label=(\roman*)]\setcounter{enumi}{2}
		\item cash additive: $\rho(X+m)=\rho(X)-m$ for all $m \in \R$ and all $X \in \mathcal{X}_{l}$.
	\end{enumerate}
\end{definition}
As is well-known, any monetary risk measure is automatically convex, see \citep{Delbaen00,fritelli03,marinacci05,drapeau01} and the references therein.

Given a loss function $l$, we define the Optimized Certainty Equivalent (OCE) introduced in \citep{ben-tal02,ben-tal01}---to which we refer for further interpretation---as follows
\begin{equation}
	\rho(X):=\inf_{\eta \in \R}\Set{E\left[ l\left( \eta-X \right) \right]-\eta}=\inf_{\eta \in \R}S_{l}\left( \eta,X \right),\quad X \in \mathcal{X}_l,
	\label{}
\end{equation}
whereby 
\begin{equation}
	S_{l}\left( \eta,X \right) :=E\left[ l\left( \eta-X \right)  \right]-\eta, \quad \eta \in \R\text{ and }X \in \mathcal{X}_{l}.
	\label{}
\end{equation}
The following proposition is known up to minor differences in the assumptions. 
See \citep{ben-tal01, fabio2012} for the case $\mathcal{X}_{l}=L^\infty$, \citep{kupper01} for the case where $l$ is differentiable, and \citep{cheridito02} for the computation of the dual representation in the general case.
For the sake of readability, we provide a short proof based on results in \citep{cheridito2008}.
\begin{proposition}\label{prop:OCE}
	Let $l$ be a loss function. 
	Then, the Optimized Certainty Equivalent on $\mathcal{X}_l$ is a lower semicontinuous cash additive risk measure taking values in $\R$.

	If, in addition, $l$ satisfies Assumption \ref{assumption}, then, for any $X \in \mathcal{X}_{l}$, there exists an optimal allocation $\eta^\ast:=\eta^\ast(X) \in \R$ such that
	\begin{equation}
		\rho\left(X\right):=E\left[ l\left( \eta^\ast -X \right) \right]-\eta^\ast
		\label{eq:optialloc}
	\end{equation}
	and this optimal allocation $\eta^\ast$ belongs to $[\essinf X, \esssup X]$ and satisfies
	\begin{equation}
	    E\left[ l^\prime_{-} \left( \eta^\ast -X \right) \right]\leq 1\leq E\left[ l^\prime_{+}\left( \eta^\ast -X \right) \right]
		\label{eq:charopteta}
	\end{equation}
	where $l^\prime_{-}$ and $l^\prime_+$ denote the left- and right-hand derivatives of $l$ respectively.
	Finally, the OCE has the representation
	\begin{equation}
		\rho(X)=\max_{Q \in \mathcal{M}_{1,l^\ast}(P)}\Set{E_{Q}\left[-X\right]-E_{P}\left[l^\ast\left(\frac{dQ}{dP}\right)  \right]}, \quad X \in \mathcal{X}_l.
		\label{eq:robrep}
	\end{equation}
	This supremum is attained for those $Q^\ast \in \mathcal{M}_{1,l^*}(P)$ where the density is such that $l_{-}^\prime(\eta^\ast-X)\leq dQ^\ast/dP\leq l_{+}^\prime\left( \eta^\ast -X \right)$, while $\eta^\ast$ fulfills \eqref{eq:charopteta}.
\end{proposition}
\begin{proof}
	Since $l(x)\geq x$ and $\mathcal{X}_{l}\subseteq L^1$, it holds $S_l(\eta,X)\geq E[-X]>-\infty$, hence $\rho(X)>-\infty$.
	On the other hand, $S_l(0,X)\leq E[l( X^-)]\leq E[l( \abs{X})]< +\infty$ since $X \in \mathcal{X}_l$.
	Hence $\rho(X)<+\infty$.

	Let us show that we have an optimal allocation determined by means of relation \eqref{eq:charopteta}.
	Given $X \in \mathcal{X}_{l}$, the function $\eta\mapsto S_l(\eta,X)$ is real-valued and convex.
	Furthermore, assumption \ref{assumption} ensures that $l(x)\geq bx +c$ and $l(x)\geq b^\prime x +c$ for all $x \in \R$ for some $b>1>b^\prime$ and $c \in \mathbb{R}$.
	Hence, it holds
	\begin{equation*}
		S_l\left(\eta,X\right)\geq E\left[ -bX+b\eta+c \right]-\eta\geq (b-1)\eta-b E[X]+c
	\end{equation*}
	which goes to $+\infty$ as $\eta$ tends to $+\infty$ since $b-1>0$.
	A similar argumentation with $b^\prime$ implies that $S_l(\eta,X)$ goes to $+\infty$ as $\eta$ tends to $-\infty$ since $b^\prime-1<0$.
	Hence, there exists a minimum $\eta^\ast \in \R$ such that \eqref{eq:optialloc} holds.
	A straightforward argumentation shows that $\eta^\ast \in [\essinf X, \esssup X]$.
	This optimal allocation fulfills the first order optimality criteria
	\[\lim_{\varepsilon \nearrow 0} \frac{S_l(\eta^\ast+\varepsilon,X)-S_l(\eta^\ast,X)}{\varepsilon}\leq 0\leq \lim_{\varepsilon \searrow 0} \frac{S_l(\eta^\ast+\varepsilon,X)-S_l(\eta^\ast,X)}{\varepsilon}.\]
	An application of Lebesgue's dominated convergence theorem allows to 
interchange limits and expectations and get relation \eqref{eq:charopteta}.

	The fact that $\rho$ is a cash additive risk measure is well-known, see \citep{ben-tal01}.
	The conditions of \citep[Theorem 2.2]{cheridito2008} are fulfilled and it holds
	\begin{equation*}
		\rho(X)=\max_{Q \in \mathcal{M}_{1,l^*}(P)}\Set{E_{Q}[-X]-\alpha(Q)}, \quad X \in \mathcal{X}_{l}
	\end{equation*}
	where 
	\begin{equation}
		\alpha(Q):=\sup_{X \in \mathcal{X}_{l}}\Set{E_{Q}[-X]-\rho(X)}, \quad Q \in \mathcal{M}_{1,l^\ast}(P).
		\label{}
	\end{equation}
	However, since $\mathcal{X}_{l}$ is a decomposable space in the sense of \citet[Definition 14.59]{rockafellar02} and $l$ is a normal integrand, we can apply \citep[Thereom 14.60]{rockafellar02} which yields 
	\begin{multline*}
		\alpha(Q)=\sup_{X \in \mathcal{X}_{l},\eta \in \R}\Set{E\left[ -\frac{dQ}{dP}X \right]+\eta-E\left[ l\left( -X+\eta \right) \right]}\\
		=\sup_{\eta\in \R}\Set{E\left[ \sup_{x \in \R}\Set{-\frac{dQ}{dP}x-l(-x+\eta)} \right]+\eta}=\sup_{\eta\in \R}\Set{E\left[ \sup_{x \in \R}\Set{\frac{dQ}{dP}\left(x-\eta\right)-l(x)} \right]+\eta}\\
		=\sup_{\eta\in \R}\Set{E\left[ l^\ast\left( \frac{dQ}{dP} \right) \right]+\eta\left(1-E\left[\frac{dQ}{dP}  \right]  \right)}=E\left[ l^\ast\left( \frac{dQ}{dP} \right) \right].
	\end{multline*}
	This shows equation \eqref{eq:robrep}.
	The representation in terms of the optimal density follows along the lines of \citep{kupper01}, by suitably adapting the proof in the case where $l$ is only convex and not necessarily differentiable.
\end{proof}

Next, we turn our attention to risk contributions, that is to the risk of 
individual factors or subportfolios of a portfolio.
The \emph{risk contribution} of a risk factor $Y$ to a portfolio $X$ is defined as follows
\begin{equation}
	RC\left( X;Y \right):=	\limsup_{\varepsilon \downarrow 0}\frac{\rho(X+\varepsilon Y)-\rho(X)}{\varepsilon}.
	\label{}
\end{equation}
In the framework of Optimized Certainty Equivalents this can also be computed 
explicitly.
\begin{proposition}\label{prop:riskcontrib}
	Let $X,Y \in \mathcal{X}_{l}$.
	If $l$ is differentiable or $X$ has a continuous distribution, then
	\begin{equation}
	    RC(X;Y)=-E\left[ Y l^\prime\left( \eta^{\ast}-X \right)\right],
	\end{equation}
	where $\eta^\ast$ satisfies $E[l^\prime(\eta^\ast-X)]=1$.
	Otherwise, we have the following bounds
	\begin{equation*}
	    E\!\left[ Y^-l_{-}^\prime\left( \eta^\ast\!-\!X \right)-Y^+l_{+}^\prime\left( \eta^{\ast}\!-\!X \right)\right]\leq RC\!\left( X;Y \right)\leq 
	    E\!\left[ Y^-l_{+}^\prime\left( \eta^\ast\!-\!X \right)-Y^+l^\prime_{-}\left( \eta^{\ast}\!-\!X \right)\right],
	\end{equation*}
	for $\eta^*$ such that 	
	$E\left[ l^\prime_{-} \left( \eta^\ast -X \right) \right]\leq 1\leq E\left[ l^\prime_{+}\left( \eta^\ast -X \right) \right]$.
\end{proposition}
\begin{proof}
	In case $l$ is differentiable and strictly convex, the proof can be found in \cite[Theorem 3.1]{kupper01}.
	Below we sketch the proof for the general case.
	Let $\eta^\ast$ be such that  $E[ l^\prime_{-}(\eta^\ast-X) ]\leq 1\leq E[l_+^\prime(\eta^\ast-X)]$, that is $\rho(X)=S_l(\eta^\ast,X)=E[l(\eta^\ast-X)]-\eta^\ast$.
	Using the convexity and monotonicity of $l$, and that $-l(x)\leq -x$, we deduce for $0<\varepsilon<1/2$ that it holds
	\begin{equation*}
		\frac{l\left( Z-\varepsilon Y \right)-l(Z)}{\varepsilon}\leq \frac{1}{1-\varepsilon}\big( l\left( Z-(1-\varepsilon)Y \big)-l\left( Z \right) \right)\leq 2\big(l\left( \abs{Z}+\abs{Y} \right)+\abs{Z}\big)\in L^1
	\end{equation*}
	for every $Z,Y \in \mathcal{X}_l$.
	Hence, by dominated convergence, it follows that
	\begin{multline*}
		\limsup_{\varepsilon \downarrow 0}\frac{\rho(X+\varepsilon Y)-\rho(X)}{\varepsilon}=\limsup_{\varepsilon \downarrow 0} \frac{\rho(X+\varepsilon Y)-S_l(\eta^\ast,X)}{\varepsilon}\\
		\leq \limsup_{\varepsilon \downarrow 0} \frac{S_l(\eta^\ast, X+\varepsilon Y)-S_l(\eta^\ast,X)}{\varepsilon}\leq E\left[\limsup_{\varepsilon \downarrow 0} \frac{l\left(\eta^\ast- X-\varepsilon Y\right)-l\left(\eta^\ast-X\right)}{\varepsilon}\right]\\
		= E\left[ Y^-l^\prime_{+}(\eta^\ast-X)-Y^+l_{-}^\prime\left( \eta^\ast-X \right) \right].
	\end{multline*}
	On the other hand, let $Z \in \mathcal{X}_{l^\ast}$ such that 
$l^\prime_-(\eta^\ast-X)\leq Z\leq l_{+}^\prime(\eta^\ast -X)$ and $E[Z]=1$.
	It follows that $ZY \leq 
Y^+l_{+}^\prime(\eta^\ast-X)-Y^{-}l_-^\prime(\eta^\ast -X)$.
	By means of \eqref{eq:robrep}, it follows that $\rho(X)=-E[ZX-l^\ast(Z)]$ and 
	\begin{multline*}
		\rho(X+\varepsilon Y)\geq -E\left[ Z(X+\varepsilon Y) \right]-E\left[ l^\ast(Z) \right]=\rho(X)-\varepsilon E\left[ ZY \right]\\
		\geq \rho(X)+\varepsilon E\left[ Y^{-}l_{-}^\prime(\eta^\ast -X)-Y^+l_{+}^\prime(\eta^\ast-X) \right]
	\end{multline*}
	Hence,
	\begin{equation*}
	    \liminf_{\varepsilon \downarrow 0} \frac{\rho\left( X+\varepsilon Y \right)-\rho\left( X \right)}{\varepsilon}\geq E\left[ Y^{-}l^\prime_{-}(\eta^\ast -X)-Y^+l_{+}^\prime(\eta^\ast-X) \right],
	\end{equation*}
	showing the bounds.
	If $l$ is differentiable then $l^\prime_{-}=l^\prime_{+}=l^\prime$ and the lower and upper bounds coincide.
	If $X$ has a continuous distribution, the set $\set{l^{\prime}_{-}( \eta^\ast-X )=l^{\prime}_{+}(\eta^\ast-X)}$ has measure one since $l^\prime_{-}$ has only countably many discontinuity points, which concludes the proof.
\end{proof}

\section{Numerical Computation of Optimized Certainty Equivalents}
\label{sec:fourier}

In this section, we develop numerical schemes for the computation of optimized certainty equivalents based on transform methods and deterministic root finding algorithms.
We also discuss the applicability of these methods for different risk scenarios and provide an example for the computation of risk contributions.
In general, the computation of the optimal allocation $\eta^*$ and the risk measure $\rho(X)$ in Proposition \ref{prop:OCE} can be performed in two steps:
\begin{enumerate}[label=\textit{Step \arabic*:},leftmargin=4em]
	\item use a deterministic root finding algorithm to compute $\eta^*$ in \eqref{eq:charopteta}, combined with transform methods for the computation of the expectations;
	\item use transform methods once more to compute the expectation $E[l(\eta^*-X)]$ and thus $\rho(X)$.
\end{enumerate}

Let $l$ be a loss function as described in the previous section and denote by $l_R$ the \emph{dampened} loss function, defined by $l_R(x):=\e^{-Rx}l(x)$, for $R\in\R$.
Moreover, let $\widehat{f}$ denote the Fourier transform of a function $f$, i.e. $\widehat{f}(u)=\int \e^{iux}f(x)\dx$, and $M_X$ the (extended) moment generating function of $X$, i.e. $M_X(u)=E[\e^{uX}]$, for suitable $u\in\C$.
By $L^1$, resp. $L^1_{\textrm{bc}}$, we denote the set of measurable functions on the real line which are integrable, resp. bounded, continuous and integrable, with respect to the Lebesgue measure.
We also denote by $K^\circ$ the interior of a set $K$ and by $\Im(z)$ the imaginary part of the complex number $z$.

The next theorem provides a general scheme for the computation of optimal allocations and risk measures in our framework following the two-step procedure described above.

\begin{theorem}\label{thm-fourier}
	Let $X \in \mathcal{X}_l$ and define
	\begin{align}
		\mathcal{I}&:=\Set{R \in \R: M_X(R)<\infty}\\
		\mathcal{J}&:=\Set{R \in \R:l_R\in L^1_{\textrm{bc}}
		\text{ and }\widehat{l}_R\in L^1}\\
		\mathcal{J}'&:=\Set{R \in \R:l_R'\in L^1_{\textrm{bc}}
		\text{ and }\widehat{l'_R}\in L^1}.
	\end{align}
	Assume that the following condition holds:
	\begin{enumerate}[label=\textbf{(A-\Roman*)},leftmargin=2cm]
		\item\label{hypoIJ} $\mathcal{I}\cap\mathcal{J}\neq\emptyset$
			\quad and \quad $\mathcal{I}\cap\mathcal{J}'\neq\emptyset$.
	\end{enumerate}
	Then the optimal allocation $\eta^*$ is the unique root of the equation $f(\eta)=0$, where
	\begin{equation}
		f(\eta) = \frac{1}{2\pi} \int_\R \e^{(R'-iu)\eta} M_X(iu-R') \widehat{l'}(u+iR') \du -1,
	\end{equation}
	with $R'\in\mathcal{I}\cap\mathcal{J}'$, and can be computed by a deterministic root finding scheme.
	Once $\eta^*$ has been determined, the risk measure has the following representation:
	\begin{equation}\label{rep:main}
		\rho(X) = \frac{1}{2\pi} \int_\R \e^{(R-iu)\eta^*} M_X(iu-R) \widehat{l}(u+iR) \du - \eta^*,
	\end{equation}
	for $R\in\mathcal{I}\cap\mathcal{J}$.
\end{theorem}
\begin{proof}
	Since we have assumed that the derivative of the loss function is continuous, \eqref{eq:charopteta} yields that the optimal allocation $\eta^*$ is the unique root of the equation $f(\eta)=0$, where
	\begin{equation*}
		f(\eta)=E[l'(\eta-X)]-1.
	\end{equation*}
	The Fourier representation of the function $f$ follows directly from \citep[Theorem 2.2]{EberleinGlauPapapantoleon08}.
	In addition, once $\eta^*$ has been computed by a deterministic root-finding algorithm, \eqref{eq:optialloc} yields that
	\begin{equation*}
		\rho(X)=E[l(\eta^*-X)]-\eta^*
	\end{equation*}
	and the Fourier representation follows again from \citep[Theorem 2.2]{EberleinGlauPapapantoleon08}.
\end{proof}

\begin{remark}
	The assumption of continuity of $l'$ can be easily relaxed by assuming more regularity of the random variable $X$;
	see the `dual' conditions in \citep[Remark 2.3]{EberleinGlauPapapantoleon08}.
	Moreover, we often divide the loss function between $\R^+$ and $\R^-$ where the two parts have different growth and regularity, and therefore we consider distinct sets $\mathcal{J}_1$ and $\mathcal{J}_2$ for each one of them.
	This is, for instance, the case in the CV@R example.
\end{remark}

The root of the equation $f(\eta)=0$ can be determined by standard root-finding
algorithms, see e.g. \citet{Stoer_Bulirsch_2002} or \citet{Press_etal_2007}. A
natural choice is to use the \textit{secant method}, where one starts with two
initial values $\eta_0,\eta_1$ such that $f(\eta_0)\neq f(\eta_1)$ and the root
$\eta^*$ is determined by 
\begin{align*}
\eta^*=\lim_{k\to\infty}\eta_k
 \quad\text{ where }\quad
\eta_{k+1} = \eta_{k} 
	- f(\eta_{k}) \cdot \frac{\eta_{k}-\eta_{k-1}}{f(\eta_{k})-f(\eta_{k-1})}.
\end{align*}
This method converges with superlinear rate if the initial values are sufficiently close to the root.
A more convenient choice is to use \textit{Brent's method}, which combines the bisection, the secant and the inverse quadratic interpolation methods; see \citet{Brent_1973} for all the details.
This method is guaranteed to converge and the rate is again superlinear (equal to $\frac{1+\sqrt{5}}{2}$) if the function is continuously differentiable near the root.
In the numerical examples, we will use Brent's method, since this is the standard root-finding algorithm implemented in Matlab.

\begin{remark}
	Although $l'$ might not be continuously differentiable (or even continuous), $f$ could still be continuously differentiable if the random variable $X$ is sufficiently regular, since the density of $X$ will `smoothen' $f$.
\end{remark}

\subsection{Explicit Fourier Representation for OCEs}\label{sec:diffloss}

In the following subsections, we provide explicit formulas for the computation of optimal allocations and OCE-based risk measures using Fourier methods and deterministic root finding schemes.
Before proceeding with examples of loss functions that fit into our framework, we will briefly review Value at Risk.  

\subsubsection{Value at Risk}

Denote the upper quantile function of the random variable $X$ by $q^+_X$, that is
\begin{equation*}
	q^+_X(u) = \inf\{x\in\R: P(X\le x)>u \}.
\end{equation*}
Then, the \textit{Value at Risk} (V@R) at some level $\lambda\in(0,1)$ is defined as
\begin{equation*}
	V@R_\lambda(X) = - q^+_X(\lambda),
\end{equation*}
see e.g. \citet[Section~4.4]{foellmer01}.
Value at Risk can be computed in a similar fashion to the OCE-based risk measures in Theorem \ref{thm-fourier}, i.e. by combining a Fourier representation for the cumulative distribution function with a root-finding algorithm.

\subsubsection{Entropic Risk Measure}

A classical example that fits in this framework is the entropic loss function
\begin{align*}
	l(x) &= \frac{\e^{\gamma x}-1}{\gamma}
\end{align*}
with $\gamma>0$.
The derivative and the conjugate functions are
\begin{align*}
	l^\prime(x) &= \e^{\gamma x} \quad\text{ and }\quad
	l^\ast(y) = \frac{y\ln (y)}{\gamma}-\frac{y-1}{\gamma}.
\end{align*}
The optimal allocation $\eta^*$ and the risk measure $\rho(X)$ can be computed explicitly and are provided by
\begin{align*}
	\eta^* &= - \frac{1}{\gamma}\ln\left( E\left[ \e^{-\gamma X} \right] \right),\\
	\rho(X)&= -\eta^\ast=\frac{1}{\gamma}\ln\left( E\left[ \e^{-\gamma X} \right] \right)\\
	&=\sup_{Q \in \mathcal{M}_{1,l^\ast}(P)}\Set{E_{Q}\left[ -X \right]-E_{Q}\left[ \ln\left( \frac{\ud Q}{\ud P} \right) \right]}.
\end{align*}
There exist many models where the moment generating function, i.e. the quantity $E[\e^{-\gamma X}]$, is known explicitly, for example L\'evy or Sato processes and affine models.
In this case, also $\eta^*$ and $\rho(X)$ can be computed explicitly.

\subsubsection{Conditional Value at Risk}

The most interesting example from the point of view of practical applications is Conditional Value at Risk, also known as Average Value at Risk or Expected Shortfall.
These notions coincide if $X$ has a continuous distribution, see \citep[Corollary 4.49]{foellmer01}.
Conditional Value at Risk is a special case of an OCE where the loss function is
\begin{equation}\label{loss_AV@R}
	l(x) = - \gamma_1x^- + \gamma_2x^+
	= \begin{cases}
		\gamma_1x &\text{ if } x\le0\\
		\gamma_2x &\text{ if } x>0,
	\end{cases}
\end{equation}
with $\gamma_2>1>\gamma_1\geq 0$.
The left-hand derivative equals
\begin{equation*}
    l_{-}^\prime(x) = \gamma_11_{\{x\leq0\}} + \gamma_21_{\{x>0\}},
\end{equation*}
while the conjugate function is
\begin{equation*}
	l^\ast(x)=
	\begin{cases}
		0\quad &\text{ if }\gamma_2\leq x\leq \gamma_1\\
		+\infty&\text{ otherwise .}
	\end{cases}
\end{equation*}
In case $\gamma_1=0$, the resulting risk measure corresponds to the standard CV@R with parameter $1/\gamma_2$, see for instance \citep{rockafellar00}.
The optimal allocation can be computed explicitly, in terms of the quantile function $q^+_X$ of $X$, and is provided by
\begin{equation}\label{opt_alloc_AV@R}
	\eta^* = q^+_X\left( \frac{1-\gamma_1}{\gamma_2-\gamma_1} \right).
\end{equation}
The following representation for this risk measure is also standard in the literature
\begin{align}\label{OCE-AV@R}
	\rho\left( X \right)
	&=-\gamma_2\int_{0}^{\frac{1-\gamma_1}{\gamma_2-\gamma_1}}q^+_{X}(s)\ds 
	-\gamma_1\int_{\frac{1-\gamma_1}{\gamma_2-\gamma_1}}^{1}q^+_{X}(s)\ds\\
	&=\sup_{Q \in \mathcal{M}_{1,l^\ast}(P)}
	\Set{E_{Q}\left[ -X \right]: \gamma_1\leq d Q/d P\leq \gamma_2 }.
\end{align}
In particular, for the special case of CV@R with parameter $\lambda=1/\gamma_2$, it holds that $\eta^*=q^+_X\left(\lambda\right)$ and
\begin{equation}\label{CV@R-std}
	CV@R_{\lambda}(X)=-\frac{1}{\lambda}\int_{0}^{\lambda}q^+_{X}(s)ds= \frac{1}{\lambda}\int_{0}^{\lambda}V@R_s(X)ds.
\end{equation}
The aim of the next result is to provide an alternative representation for $\rho(X)$ using Fourier transform methods.
\begin{proposition}\label{prop-AV@R}
	Assume that the optimal allocation $\eta^*$ is computed by \eqref{opt_alloc_AV@R}.
	Let $X\in L^0$ be a random variable such that $0\in\mathcal{I}^\circ$.
	Then $X \in L^1=\mathcal{X}_l$ and the risk measure $\rho$ admits the following representation
	\begin{equation}\label{OCE-AV@R-fourier}
		\rho(X)= \frac{\gamma_1}{2\pi} \int_\R \frac{\e^{(R_1-iu)\eta^*}}{(u+iR_1)^2}M_X(iu-R_1) \du - \frac{\gamma_2}{2\pi} \int_\R \frac{\e^{(R_2-iu)\eta^*}}{(u+iR_2)^2}M_X(iu-R_2) \du- \eta^*,
	\end{equation} 
	where $R_1\in\mathcal{I}\cap(-\infty,0)$ and $R_2\in\mathcal{I}\cap(0,+\infty)$.
	In particular, for CV@R we get
	\begin{equation}\label{CV@R-fourier}
		CV@R_{\lambda}(X)= -\frac{1}{2\pi\lambda} \int_\R \frac{\e^{(R-iu)\eta^*}}{(u+iR)^2}M_X(iu-R) \du	- \eta^*,
	\end{equation}
	where $\lambda=1/\gamma_2$ and $R \in \mathcal{I}\cap (0,+\infty)$.
\end{proposition}
\begin{proof}
	The loss function grows linearly while $X$ has finite exponential 
moments, thus $\mathcal{X}_l=L^1$ and $X\in L^1$.
	Since $\eta^*$ is already computed using the first order condition \eqref{eq:charopteta} for the loss function \eqref{loss_AV@R}, see	\eqref{opt_alloc_AV@R}, we will apply the second part of Theorem \ref{thm-fourier} directly to representation \eqref{eq:optialloc}.
	We get
	\begin{align}\label{Fourier-CV@R-rep}
		\rho(X) 
		&= E\left[l(\eta^*-X)\right] - \eta^*
		= E\left[- \gamma_1(\eta^*-X)^- + \gamma_2(\eta^*-X)^+ \right]- \eta^* \nonumber\\
		&= -\gamma_1E\left[(X-\eta^*)^+\right] + \gamma_2 E\left[(\eta^*-X)^+\right]- \eta^* \nonumber\\
		&= -\frac{\gamma_1}{2\pi} \int_\R \e^{(R_1-iu)\eta^*} M_X(iu-R_1) \widehat{l}_1(u+iR_1) \du \nonumber\\
		&\qquad+ \frac{\gamma_2}{2\pi} \int_\R \e^{(R_2-iu)\eta^*} M_X(iu-R_2)\widehat{l}_2(u+iR_2) \du
		- \eta^*,
	\end{align}
	for $R_1\in\mathcal{I}\cap\mathcal{J}_1$, $R_2\in\mathcal{I}\cap\mathcal{J}_2$, where we define the functions
	\begin{equation*}
		l_1(x) = (-x)^+
		\quad\text{ and }\quad
		l_2(x) = (x)^+.
	\end{equation*}
	Now, we just have to compute the Fourier transforms of the functions $l_1$ and $l_2$, determine the sets $\mathcal{J}_1$ and $\mathcal{J}_2$, and show that the prerequisites of Theorem \ref{thm-fourier} are	satisfied.
	
	The Fourier transform of $l_1$, for $z\in\C$ with $\Im(z)\in(-\infty,0)$, is provided by
	\begin{equation}\label{Fourier-AVaR-fc}
		\widehat{l}_1(z)
		= \int_\R \e^{izx}(-x)^+ \dx
		= -\int_{-\infty}^0 \e^{izx}x \dx
		= -\frac{x}{iz}\e^{izx}\Big|_{-\infty}^0 + \frac{1}{(iz)^2}\e^{izx}\Big|_{-\infty}^0
		= - \frac{1}{z^2},
	\end{equation}
	while for $l_2$ we get the same formula, that is
	\begin{equation}\label{Fourier-AVaR-fp}
		\widehat{l}_2(z) = - \frac{1}{z^2},
	\end{equation}
	where now $z\in\C$ with $\Im(z)\in(0,+\infty)$.
	The corresponding dampened payoff functions are
	\begin{equation}
		l_{1,R_1}(x) = \e^{-R_1 x}(-x)^+\,
		\quad\text{and}\quad
		l_{2,R_2}(x) = \e^{-R_2 x}(x)^+.
	\end{equation}
	Clearly, for $R_1<0$ and $R_2>0$, these functions are bounded and continuous, while from \eqref{Fourier-AVaR-fc} it directly follows that $l_{1,R_1},l_{2,R_2}\in L^1$.
	A direct computation shows that also $l_{1,R_1},l_{2,R_2}\in L^2$.
	Indeed,
	\begin{equation}
		\norm{l_{1,R_1}}^2_{L^2}
		= \int_\R \abs{l_{1,R_1}(x)}^2 \dx
		= \int_{-\infty}^0 \e^{-2R_1 x} x^2 \dx
		= \frac{1}{4R_1^3} <\infty,
	\end{equation}
	while the computation for $l_{2,R_2}$ is completely analogous.
	We can also examine the weak derivatives of $l_{1,R_1},l_{2,R_2}$; we get that
	\begin{equation}
		\partial l_{1,R_1}(x)
		= \begin{cases}
			\e^{-R_1 x}(R_1x-1)&\text{for }x>0 \\
			0&\text{for }x<0
		\end{cases}
	\end{equation}
	from which we can directly deduce that $\partial l_{1,R_1}\in L^2$, while the same is true for $\partial l_{2,R_2}$.
	Thus, $l_{1,R_1}$, $l_{2,R_2}$ belong to the Sobolev space
	\begin{align}
		H^1(\R)=\Set{ g\in L^2 :\partial g \text{ exists and }\partial g\in L^2}
	\end{align}
	and using \citep[Lemma~2.5]{EberleinGlauPapapantoleon08} we can conclude that	$\widehat{l}_{1,R_1},\widehat{l}_{2,R_2}$ are integrable.
	Therefore, $\mathcal{J}_1=(-\infty,0)$ and $\mathcal{J}_2=(0,+\infty)$.

	Finally, since $0\in\mathcal{I}^\circ$, we have that $\mathcal{I}\cap\mathcal{J}_1\neq\emptyset$ and $\mathcal{I}\cap\mathcal{J}_2\neq\emptyset$, hence assumption \ref{hypoIJ} is satisfied.
	The result now follows by substituting \eqref{Fourier-AVaR-fc}--\eqref{Fourier-AVaR-fp} into \eqref{Fourier-CV@R-rep}.
\end{proof}

\subsubsection{Polynomial Loss Function}

Another interesting example is the class of polynomial loss functions.
The polynomial loss function is defined by
\begin{equation}\label{poly-lf}
	l(x) = \frac{\left( \left[1+x\right]^+ \right)^\gamma-1}{\gamma}
\end{equation}
for $\gamma\in\N$, $\gamma>1$.
The case $\gamma=2$ corresponds to the Monotone Mean-Variance, cf. \citep{fabio2012}.
The derivative equals
\begin{equation}
	l^\prime(x) = \left(\left[1+x\right]^+\right)^{\gamma-1},
\end{equation}
and the conjugate function is provided by
\begin{equation}
	l^{\ast}(y)=
	\begin{cases}
		\displaystyle \frac{(1-\gamma)y^{\frac{\gamma}{\gamma-1}}-\gamma y-1}{\gamma}&\text{ if }y\geq 0\\
		+\infty&\text{otherwise}.
	\end{cases}
\end{equation}
In this class of loss functions neither the optimal allocation nor the OCE can be computed explicitly and one has to resort to numerical methods for both.
\begin{proposition}
	Let $X\in L^0$ be a random variable such that $0\in\mathcal{I}^\circ$. 
	Then $X\in\mathcal{X}_l=L^\gamma$ and the optimal allocation is the unique solution of the equation $f(\eta)=0$ where
	\begin{equation}\label{secant-poly}
		f(\eta) = \frac{(\gamma-1)!}{2\pi} \int_\R M_X(iu-R) \frac{\e^{(R-iu)(1+\eta)}}{(R-iu)^\gamma} \du - 1,
	\end{equation}
	with $R\in\mathcal{I}\cap(0,+\infty)$.
	Once $\eta^*$ is determined, the polynomial loss function risk measure admits	the following representation
	\begin{equation}\label{OCE-poly}
		\rho(X) =  \frac{(\gamma-1)!}{2\pi} \int_\R M_X(iu-R) \frac{\e^{(R-iu)(1+\eta^*)}}{(R-iu)^{\gamma+1}} \du - \frac1\gamma - \eta^*.
	\end{equation}
\end{proposition}
\begin{proof}
	We start by computing the Fourier transform of the following function:
	\begin{equation}
		\varphi(x) = [(a+x)^+]^n
	\end{equation}
	for $a\in\R, n\in\N$.
	Integrating by parts iteratively we get, for $z\in\C$ with $\Im(z)\in(0,+\infty)$, that
	\begin{align}\label{fourier-poly}
		\widehat{\varphi}(z)
		&= \int_\R\e^{izx}\varphi(x)dx
		 = \int_{-a}^{\infty} \e^{izx}(a+x)^ndx \nonumber\\
		&= \underbrace{\frac{\e^{izx}}{iz}(a+x)^n\Big|_{-a}^{\infty}}_{=0}
		 - \frac{n}{iz}\int_{-a}^{\infty} \e^{izx}(a+x)^{n-1}dx \nonumber \, 
		 = \dots \nonumber\\
		&= (-1)^n\frac{n!}{(iz)^n}\int_{-a}^{\infty} \e^{izx}dx
		 = n!\left(\frac{i}{z}\right)^{n+1}.
	\end{align}
	Following the same argumentation as in the proof of Proposition \ref{prop-AV@R},	we can show that the dampened function $\varphi_R$ belongs to $L^1_{\textrm{bc}}$ and has an integrable Fourier transform for $R\in(0,+\infty)$. 

	Clearly $\mathcal{X}_l=L^\gamma$ and $X\in L^\gamma$. Now consider the 
function
	\begin{align}\label{f-polynomial}
		f(\eta) 
		= E\big[l'(\eta-X)\big]-1
		= E\left[\left([1+\eta-X]^+\right)^{\gamma-1}\right]-1.
	\end{align}
	According to \eqref{eq:charopteta}, the zero of this function determines the optimal allocation corresponding to the polynomial loss function	\eqref{poly-lf}, and this can be determined by a standard root-finding algorithm.
	Applying Theorem \ref{thm-fourier} to \eqref{f-polynomial}, using	\eqref{fourier-poly} with $a=1,n=\gamma-1$, and recalling that $\mathcal{I}\cap\mathcal{J}'\neq\emptyset$ since $0\in\mathcal{I}^\circ$ and $\mathcal{J}'=(0,+\infty)$, yields the representation \eqref{secant-poly}.

	Once the optimal allocation has been determined numerically, we just have to combine \eqref{eq:optialloc}, \eqref{poly-lf}, Theorem \ref{thm-fourier}, and \eqref{f-polynomial} with $a=1,n=\gamma$, and direct computations yield representation \eqref{OCE-poly} for the risk measure corresponding to the polynomial loss function.
\end{proof}

\subsection{Scenarios and Computation of Risk Contributions}

The framework we consider is very flexible because on the one hand it accommodates a variety of different loss functions, while on the other hand the only information needed about the underlying risk factor $X$ is its moment generating function.
This is the reason why a variety of different scenarios can be treated simultaneously:
\begin{enumerate}[label=\textit{S\arabic*:}]
	\item The risk factor corresponds to an asset or a portfolio with known moment generating function (e.g. estimated from market data).
	\item The risk factor corresponds to the random claims against an insurer, that is $X=\sum_{i=1}^N X_i$, where $N,X_1,X_2,\dots$ are independent and $N$ takes values in $\mathbb{N}_0$.
		Then it holds
		\begin{equation*}
			M_X(u)=P(N=0)+\sum_{i=1}^\infty P(N=i)\prod_{j=1}^i M_{X_j}(u).
		\end{equation*}
		A weighted portfolio of financial assets, that is, $X=\sum_{i=1}^N w_iX_i$, where $N$ is fixed, $w_i$ is deterministic and $X_i$, $i=1,\dots,N$, are independent, can be treated analogously.
	\item The risk factor describes the total loss of a portfolio in the 
spirit of \citet{DDD}, that is $X=\sum_{i=1}^n Z_i U_i$, where $n$ is a finite 
number of financial positions, $U_i\ge 0$, and $Z_1,U_1,\dots, Z_n, U_n$ are 
independent.
		The random variable $Z_i$ determines whether $i$ defaults $(Z_i=1)$ or not $(Z_i=0)$, and $U_i$ determines the exposure at default.
		In that case
		\begin{equation*}
			M_X(u)=E\left[\exp\left(u\sum_{i=1}^n Z_i U_i\right)\right]
			      =\prod_{i=1}^n\big\{ P(Z_i=0)+P(Z_i=1)M_{U_i}(u) \big\}.
		\end{equation*}
	\item An easy and popular way to generate dependence is using a linear mixture model (cf. e.g. \citet{MadanKhanna09}, and \citet{Kawai09}).
	      Let $Y_1,\dots,Y_m$  be independent random variables, then the dependent factors $U=(U_1,\dots,U_n)$ can be defined via $U=AY$ for $A\in\R^{n\times m}$.
	      Assuming that the moment generating function of the $Y_i$'s is known, the moment generating function of the risk factor $X=\sum_{i=1}^nU_i$ is provided by
	      \begin{equation*}
			M_X(u) = \prod_{l=1}^{m} M_{Y_l}\left(u\alpha_l\right),
	      \end{equation*}
	      where $\alpha_l:=\sum_{i=1}^{n}A_{il}$.
\end{enumerate}

\subsubsection{Risk Contribution}

We present below an example where the risk contribution is computed explicitly using Fourier methods.
Let $X=\sum_{i=1}^n X_i$ be a portfolio, where $X_1,\ldots,X_n$ are independent random variables with continuous joint distribution.
We are interested in the $CV@R_{\lambda}$-risk contribution of the risk factor 
$Y=X_\ell$, $1\le\ell\le n$, to this portfolio, that is computing 
$RC(X;Y)=RC(\sum_{i=1}^n X_i;X_\ell)$ in the case where $l(x)=\gamma_2 x^+$ 
with $\gamma_2=1/\lambda$.

In order to compute this risk contribution, we will make use of the following notation.
Define the random vector
\begin{align*}
	(Z,Y) := \Big(\sum_{\substack{i=1\\ i\neq \ell}}^nX_i,X_\ell\Big)
\end{align*}
and denote its probability measure by $P_{Z,Y}$ and its moment generating function by $M_{Z,Y}$.
Moreover, define the measure $\varrho_R(dx):=\e^{\langle R,x\rangle}P_{Z,Y}(dx)$ and introduce the sets
\begin{align}
	\mathcal{Y} := \Set{R\in\R^2: M_{Z,Y}(R)<\infty},
\end{align}
and
\begin{align}
	\mathcal{Z} := \Set{R\in\R^2: \widehat\varrho_R\in L^1}.
\end{align}

\begin{proposition}\label{contribution-fourier}
	Let $X,X_\ell\in L^1$ and assume that $\mathcal{Y}\cap\mathcal{Z}\neq\emptyset$. 
	Moreover, assume that the optimal allocation $\eta^*$ has been computed by \eqref{opt_alloc_AV@R}.
	Then, the risk contribution admits the following representation:
	\begin{multline}
		RC(X;X_\ell)
		=\frac{\gamma_2}{4 \pi^2}\int_{\R^2}^{}M_{Z,Y}(R+iu)\frac{\e^{-(R_1+iu_1)\eta^\ast}}{(R_1+iu_1)(u_1-u_2-iR_1-iR_2)^2}du\\
		+\frac{\gamma_2}{4 \pi^2}\int_{\R^2}^{}M_{Z,Y}(R^\prime+iu)\frac{\e^{-(R_1^\prime+iu_1)\eta^\ast}}{(R^\prime_1+iu_1)(u_1-u_2-iR^\prime_1-iR^\prime_2)^2}du,
		\label{}
	\end{multline}
	where
	\begin{equation}
		M_{Z,Y}(u_1,u_2)=M_{X_\ell}(u_2)\prod_{\substack{i=1\\ i\neq\ell}}^n M_{X_i}(u_1),
		\label{mgf-contrib-fourier}
	\end{equation}
	for $R,R'\in\mathcal{Y}\cap\mathcal{Z}$ such that $R_1<0$, $R^\prime_1<0$, $R_1+R_2<0$ and $R^\prime_1+R^\prime_2>0$.
\end{proposition}
\begin{proof}
	Using Proposition \ref{prop:riskcontrib} with $l(x)=\gamma_2 x^+$, it follows directly that
	\begin{align*}
		RC(X;X_\ell)
		&=-\gamma_2E\left[ Y1_{\set{\eta^\ast>X}} \right]\\
		&=-\gamma_2E\left[ Y1_{\set{Y\leq 0}}1_{\set{\eta^\ast> Z+Y}}+Y1_{\set{Y>0}}1_{\set{\eta^\ast > Z+Y}} \right]\\
		&=-\gamma_2 E\big[ \psi_1(Z,Y) + \psi_2(Z,Y) \big]\\
		&=-\frac{\gamma_2}{4\pi^2}\left(  \;\int_{\R^2}^{}\!M_{Z,Y}(R+iu)\widehat{\psi}_1(iR-u)du \!+\! \int_{\R^2}^{}\!M_{Z,Y}(R^\prime+iu)\widehat{\psi}_2(iR^\prime-u)du \right)
	\end{align*}
	where $\psi_1(z,y):=y1_{\set{y\leq 0}}1_{\set{\eta^\ast> z+y}}$ and $\psi_2(z,y):=y1_{\set{y> 0}}1_{\set{\eta^\ast> z+y}}$.
	The last equality follows from \citep[Theorem 3.2]{EberleinGlauPapapantoleon08}, noting that for $R,R'\in\mathcal{Y}\cap\mathcal{Z}$ assumptions (A2) and (A3) therein are satisfied.
	
	Now, by independence we get immediately, for $u\in\mathcal{Y}$, that
	\begin{equation*}
		M_{Z,Y}(u)=E\left[ \e^{u_1Z+u_2Y} \right]=E\left[ \e^{u_1 \sum_{i=1,i\neq \ell}^n X_i+u_2X_\ell}\right]=M_{X_\ell}(u_2)\prod_{\substack{i=1\\ i\neq \ell}}^n M_{X_i}(u_1).
	\end{equation*}
	Next, we have to compute the Fourier transforms of the functions $\psi_1$ and $\psi_2$.
	We have, for $u\in\C$ with $\Im(u_1)<0$ and $\Im(u_2-u_1)<0$,
	\begin{align*}
		\widehat{\psi}_1(u)
		&=\int_{\R^2}^{}\e^{iu_1z+iu_2y}\psi_1(z,y)dzdy
		 =\int_{-\infty}^{0}\int_{-\infty}^{\eta^\ast-y}\e^{iu_1 z+iu_2y}y dzdy\\
		&=\frac{\e^{iu_1\eta^*}}{iu_1}\int_{-\infty}^{0}\e^{i(u_2-u_1)y}ydy
		 =\frac{\e^{iu_1\eta^*}}{iu_1(u_2-u_1)^2}.
	\end{align*}
	Similarly, for $u\in\C$ with $\Im(u_1)<0$ and $\Im(u_2-u_1)>0$, we get that the Fourier transform of $\psi_2$ equals
	\begin{align*}
		\widehat{\psi}_2(u)
		 =\frac{\e^{iu_1\eta^*}}{iu_1(u_2-u_1)^2},
	\end{align*}
	while we can easily observe that assumption (A1) from \citep[Theorem 3.2]{EberleinGlauPapapantoleon08} is also satisfied.
	Finally, the proof is completed by putting the pieces together.
\end{proof}

\begin{remark}
While the portfolio $X$ in this example contains $n$ variables, we can compute the risk contribution using only a 2-dimensional numerical integration, since only two variables are important: $Z$ and $Y$.
The same is true if we are interested in the contribution of a subportfolio $Y=\sum_{i=1}^mX_i$, $m<n$, to the total portfolio $X$.
On the contrary, the Monte Carlo computation of the risk contribution would require the simulation of all $n$ variables and thus is significantly more time consuming.
\end{remark}

\begin{remark}
Consider the scenario \textit{S4} with dependent risks, and assume we want to compute the contribution of a risk factor $U_\ell$ to the total portfolio $X=\sum_{i=1}^nU_i$. 
Then, we can apply Proposition \ref{contribution-fourier} directly, by just replacing the moment generating function in \eqref{mgf-contrib-fourier} with
\begin{equation}
M_{Z,Y}(u)=\prod_{k=1}^m M_{Y_k}\left(u_1\beta_k+u_2A_{\ell k}\right),
\end{equation}
where $\beta_k=\sum_{i=1,i\neq\ell}^nA_{ik}$.
\end{remark}

\section{Numerical Analysis and Examples}
\label{sec:numerics}

The aim of this section is to analyze and test the numerical methods for the computation of risk measures developed in the previous section.
We start by considering scenario \textit{S1} and assuming that the risk factor $X$ has a known distribution and moment generating function.
We consider the normal inverse Gaussian (NIG) distribution, which is very flexible and exhibits a variety of behaviors ranging from fat-tails to high peaks.
This distribution has been extensively studied as a model for financial markets both under the real-world and under the risk-neutral measure; see e.g. \citet{EberleinPrause02}, \citet{Barndorff-NielsenPrause01}, and \citet{Schoutens03}.
The NIG distribution has four parameters and the parameter space is $\alpha>0$, $0\leq|\beta|<\alpha$, $\delta>0$ and $\mu\in\mathbb{R}$.
The moment generating function of the NIG distribution has the following form
\begin{align}
	M_X(u) = \exp\left(u\mu 
	+ \delta\big[\sqrt{\alpha^2-\beta^2}-\sqrt{\alpha^2-(\beta+u)^2}\big]\right).
\end{align}
It is well-defined for $u\in(-\alpha-\beta,\alpha-\beta)=:\mathcal{I}$ and $0\in\mathcal{I}^\circ$. 
The density and other quantities of interest, e.g. mean and variance, can be found in \citet{Eberlein01a} or \citet{BarndorffNielsen98}.
The parameters have roughly the following impact on the shape of the density:
\begin{itemize}
	\item $\alpha$ is a shape parameter and determines the heaviness of the tails and the height of the peak;
	\item $\beta$ is a skewness parameter;
	\item $\delta$ is a scaling parameter and determines the variance;
	\item $\mu$ is a location parameter.
\end{itemize}
See Figure \ref{NIG_test} for a graphical illustration of the impact of the parameters $\alpha,\beta$ and $\delta$ on the shape of the density. 
\begin{figure}[ht!]
	\centering
	\includegraphics[width=4.8cm]{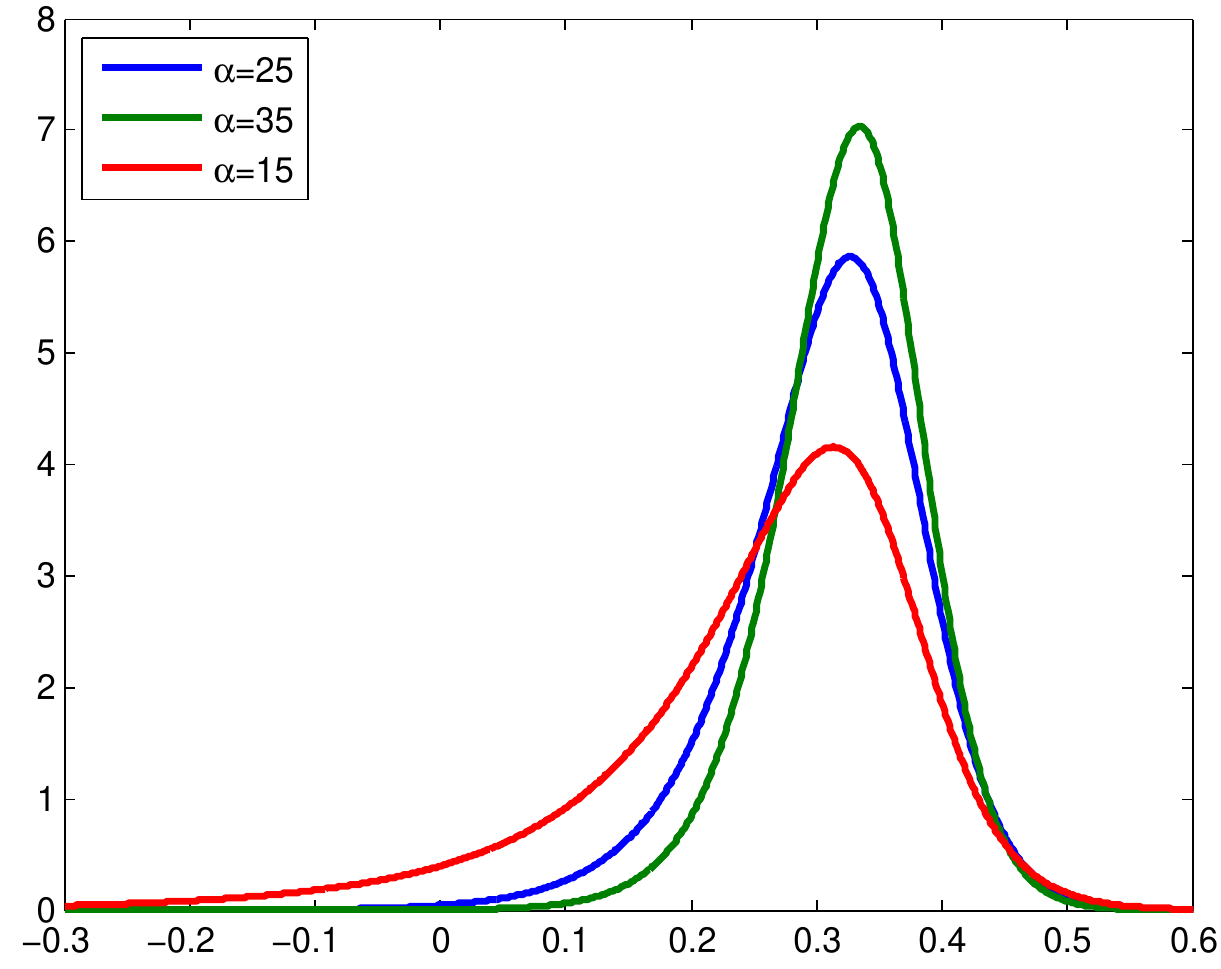}
	\includegraphics[width=4.8cm]{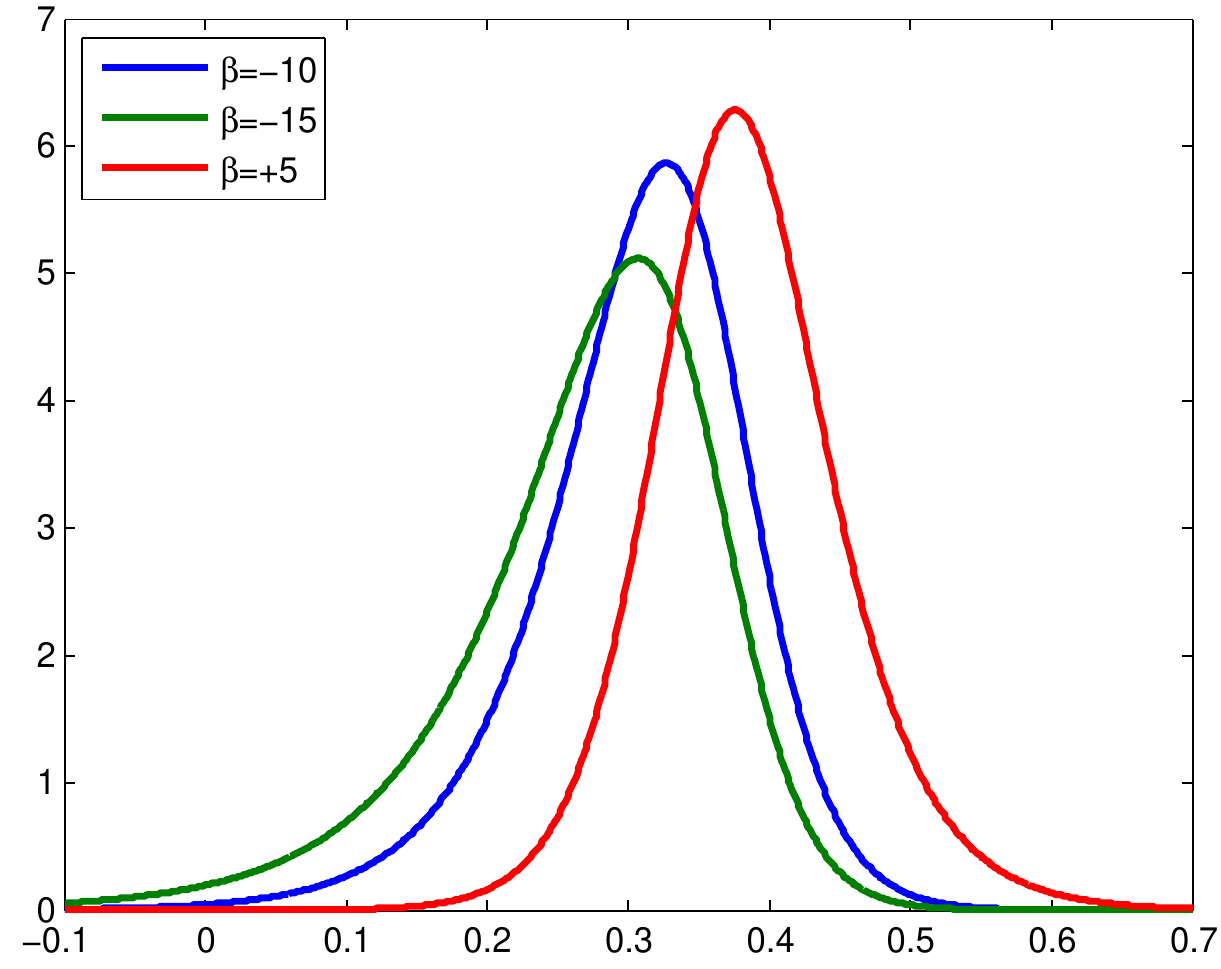} 
	\includegraphics[width=4.8cm]{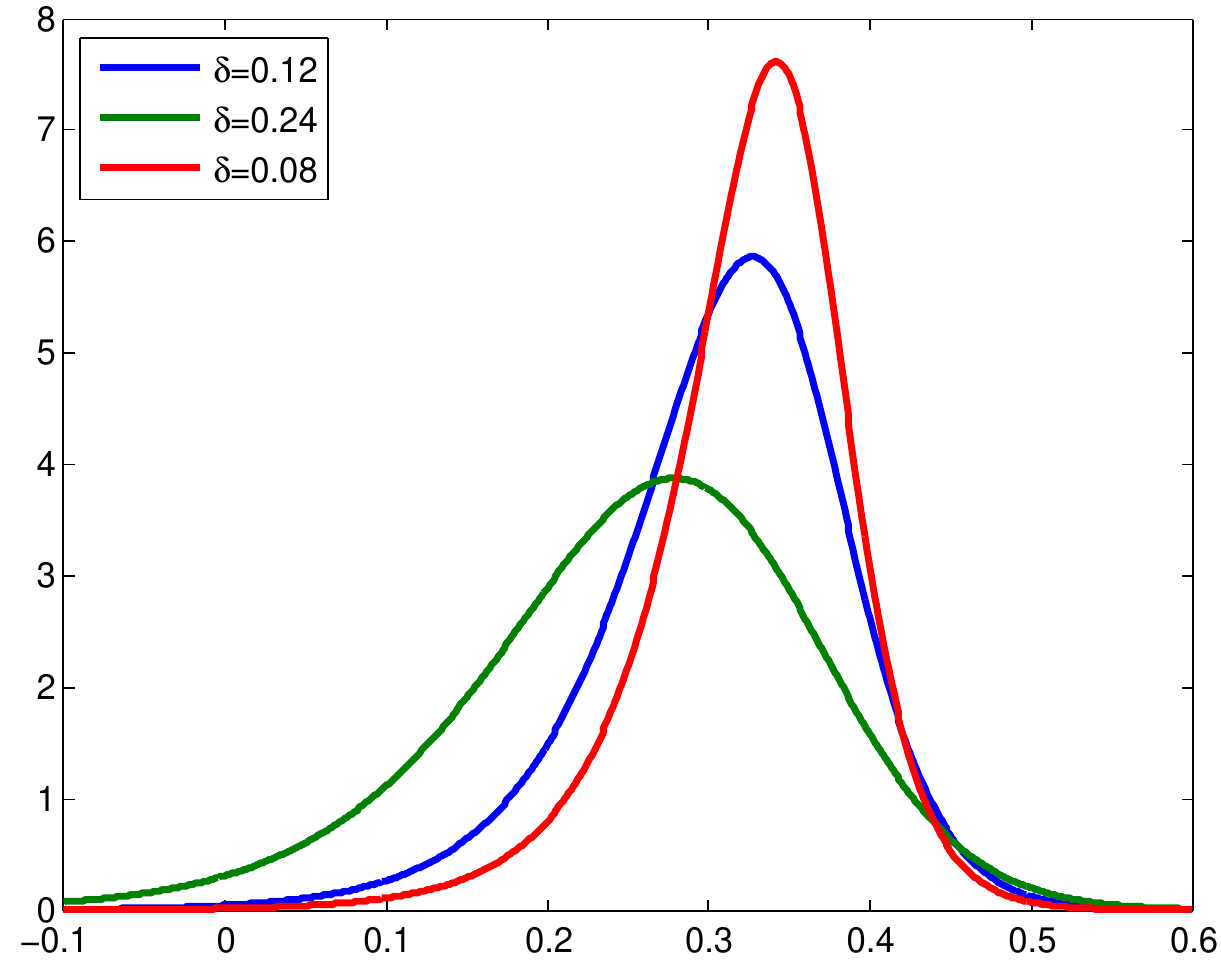}
	\caption{NIG densities with varying $\alpha,\beta$ and $\delta$.}
	\label{NIG_test}
\end{figure}

In order to make the numerical examples realistic we consider parameter sets for the NIG distribution stemming from real data.
The four different sets we consider are summarized in Table \ref{NIG_sets}, and correspond to parameters estimated from daily and monthly returns, and from options data; cf. \citep{EberleinPrause02,Schoutens03}.
Only the last set of parameters is artificial, and corresponds to a random variable with heavy tails, zero mean and variance one.
These parameters exhibit a smooth transition from densities with high peaks to densities with fat tails, and serve to test the numerical methods in a variety of different situations.
We have set $\mu=0$ in all cases, since this is completely irrelevant for the computation of risk measures.
\begin{table}[h!]
	\begin{center}
		\begin{tabular}{@{}lrrr@{}}
			\toprule
			&\multicolumn{3}{c}{Parameters}\\
			\cmidrule{2-4}
			& \multicolumn{1}{c}{$\alpha$} & \multicolumn{1}{c}{$\beta$} & \multicolumn{1}{c}{$\delta$} \\
			 \midrule
			$\text{NIG}_1$ & 106.00        &    -26.00    &    0.0110\\ 
			$\text{NIG}_2$ &  26.00        &    -10.60    &    0.0070\\ 
			$\text{NIG}_3$ &   6.20        &     -3.90    &    0.0011\\ 
			$\text{NIG}_4$ &   1.00        &      0.00    &    1.0000\\ 
			\bottomrule
		\end{tabular} 
		\caption{Parameters sets for NIG distributions.}
		\label{NIG_sets}
	\end{center}
\end{table}

\subsection{CV@R}
We want to compare here the Fourier representation \eqref{CV@R-fourier} for the Conditional Value at Risk with the standard representation \eqref{CV@R-std}.
A careful observation of these two formulas reveals that the Fourier representation should be numerically more efficient than the standard one.
Indeed, while the latter requires to solve an optimization problem---the computation of the quantile $q^+_X$---for every grid point used in the numerical integration, the former requires to solve \emph{only one} optimization problem for the computation of $\eta^*$.
Let us assume that the grid for the numerical integration has size $N$, the computational effort for the solution of the optimization problem is $M_O$, while the computational effort for the numerical integration is $M_I$, where typically $M_I\ll M_O$.
Then, the total computational effort (TCE) for the two methods compares as follows:
\begin{align}
	\mathrm{TCE(Fourier)} \cong M_O+M_I 
	\quad \text{vs} \quad
	\mathrm{TCE(standard)} \cong N\cdot M_O+M_I.
\end{align}
This also reveals that the computation of CV@R should \textit{not} be 
significantly more time consuming than the computation of V@R when the Fourier 
representation is used.
Indeed, the bulk of the computation amounts to the solution of the optimization problem (for the quantile or V@R) and not to the numerical integration.

We have computed CV@R using the Fourier and the standard representation for the four parameters sets described in Table \ref{NIG_sets}, at the $\lambda=5\%$ and the $\lambda=1\%$ level. 
The results are reported in Tables \ref{NIG_CVaR_5} and \ref{NIG_CVaR_1} respectively. 
We have also computed V@R for the same levels.
The implementation was done in Matlab and for the computation of the quantile we have used an existing package for the NIG distribution, while the results have been verified with Python and R.
The tables contain the values of V@R and CV@R, the computational time for V@R (CT), and the computational times for CV@R with the Fourier (CT(F)) and the standard representation (CT(S)).

The numerical results are completely in accordance with the analysis above.
Indeed, we can immediately observe that the computational times for CV@R using the standard representation are significantly longer than the corresponding times for the Fourier alternative.
The factor of this difference is at least equal to two, while it equals seven for the third set at the 5\% level.
In addition, we can also observe that the computational times for CV@R using the Fourier method are only marginally longer than the respective times for the computation of V@R.
This value is typically a few thousandths of a second.
This last observation should be an argument in favor of using CV@R for practical applications.

\begin{table}[h!]
	\begin{center}
		\begin{tabular}{@{}lrrcrrr@{}}
			\toprule
			&\multicolumn{2}{c}{V@R}&&\multicolumn{3}{c}{CV@R}\\
			\cmidrule{2-3}\cmidrule{5-7}
			&\multicolumn{1}{c}{Value} & \multicolumn{1}{c}{CT} &&   \multicolumn{1}{c}{Value} & \multicolumn{1}{c}{CT (F)} & \multicolumn{1}{c}{CT (S)} \\ 
			\midrule
			$\text{NIG}_1$ & 0.0210 & 0.092 && 0.0298 &   0.099         & 0.212\\ 
			$\text{NIG}_2$ & 0.0311 & 0.087 && 0.0585 &   0.094         & 0.359\\ 
			$\text{NIG}_3$ & 0.0073 & 0.088 && 0.0352 &   0.097         & 0.636\\ 
			$\text{NIG}_4$ & 1.5914 & 0.089 && 2.2872 &   0.097         & 0.197\\ 
			\bottomrule
		\end{tabular} \\
		\caption{Numerical results for V@R and CV@R at the 5\% level. Time in seconds.}
		\label{NIG_CVaR_5}
	\end{center}
\end{table}

\begin{table}[h!]
	\begin{center}
		\begin{tabular}{@{}lrrcrrr@{}}
			\toprule
			&\multicolumn{2}{c}{V@R}&&\multicolumn{3}{c}{CV@R}\\
			\cmidrule{2-3}\cmidrule{5-7}
			&\multicolumn{1}{c}{Value} & \multicolumn{1}{c}{CT} &&   \multicolumn{1}{c}{Value} & \multicolumn{1}{c}{CT (F)} & \multicolumn{1}{c}{CT (S)} \\ 
			\midrule
			$\text{NIG}_1$ & 0.0350 & 0.095 && 0.0444 &           0.104 & 0.211\\ 
			$\text{NIG}_2$ & 0.0737 & 0.092 && 0.1108 &           0.099 & 0.360\\ 
			$\text{NIG}_3$ & 0.0369 & 0.088 && 0.1162 &           0.100 & 0.507\\ 
			$\text{NIG}_4$ & 2.7019 & 0.094 && 3.4503 &           0.099 & 0.194\\ 
			\bottomrule
		\end{tabular} \\
		\caption{Numerical results for V@R and CV@R at the 1\% level. Time in seconds.}
		\label{NIG_CVaR_1}
	\end{center}
\end{table}

\begin{remark}
In case the risk factor $X$ has a known density function (scenario \textit{S1}),
as is the case for the normal inverse Gaussian distribution, we can directly
integrate over the density to compute CV@R. We have the following
representation 
\begin{align}\label{CV@R-density}
CV@R_\lambda(X) = \frac{1}{\lambda} \int_\R (\eta^*-x)f_X(x)dx - \eta^*,
\end{align}
where $f_X$ denotes the density of the random variable $X$. We have tested this
method numerically and, while it yields very competitive---in terms of
computational times---results for the third and fourth datasets, it fails
completely for the first and second datasets. The reason is that these data
correspond to densities with very high peaks and small variance, and the
standard discretization in Matlab is not sufficient to deliver the correct
values. Since these datasets correspond to 1-day and 1-month returns, while in
practice risk measures for 10-days returns have to be computed, one should be
very careful when using \eqref{CV@R-density}.
\end{remark}

\subsection{Polynomial Risk Measures}

In the last numerical experiment we compute polynomial risk measures using the Fourier methodology developed here.
We consider again scenario \textit{S1} and use the parameters for the normal inverse Gaussian distribution from Table \ref{NIG_sets}.
We consider three exponents for the relative risk aversion parameter: $\gamma=2$ which corresponds to monotone mean-variance, $\gamma=4$ which corresponds to quartic utility (cf. \citet{Hamm_Salfeld_Weber_2012}) and $\gamma=5$.
We have first computed the optimal allocation using representation \eqref{secant-poly} in combination with Brent's root finding algorithm and then calculated the corresponding risk measure using \eqref{OCE-poly}.
The values of both $\eta^*$ and $\rho(X)$ for all datasets and exponents are reported in Tables \ref{Table_poly_2} and \ref{Table_poly_4-5} together with the respective computational times for the Fourier representation (CT(F)).
  
\begin{table}[h!]
	\begin{center}
		\begin{tabular}{@{}lrrcrrr@{}}
			\toprule
			&\multicolumn{3}{c}{Fourier}&&\multicolumn{2}{c}{SRF}\\
			\cmidrule{2-4}\cmidrule{6-7}
			&\multicolumn{1}{c}{$\eta^*$} & \multicolumn{1}{c}{$\rho(X)$} & \multicolumn{1}{c}{CT(F)} && \multicolumn{2}{c}{CT} \\ 
			\midrule
			$\text{NIG}_1$ & -0.0028 &  0.0028 & 0.045 &  & 0.455 &\\ 
			$\text{NIG}_2$ & -0.0031 &  0.0033 & 0.051 &  & 0.449 &\\ 
			$\text{NIG}_3$ & -0.0009 &  0.0011 & 0.102 &  & 0.443 &\\ 
			$\text{NIG}_4$ & -0.0957 &  0.4380 & 0.044 &  & 0.448 &\\ 
			\bottomrule
		\end{tabular} \\
		\caption{Polynomial risk measure with $\gamma=2$. Time in seconds.}
		\label{Table_poly_2}
	\end{center}
\end{table} 

\begin{table}[h!]
	\begin{center}
		\begin{tabular}{@{}lrrrcrrrr@{}}
			\toprule
			&\multicolumn{3}{c}{$\gamma=4$}&&\multicolumn{3}{c}{$\gamma=5$}\\
			\cmidrule{2-4}\cmidrule{6-8}
			&\multicolumn{1}{c}{$\eta^*$} & \multicolumn{1}{c}{$\rho(X)$} & \multicolumn{1}{c}{CT(F)} &&\multicolumn{1}{c}{$\eta^*$} & \multicolumn{1}{c}{$\rho(X)$} & \multicolumn{1}{c}{CT(F)} \\ 
			\midrule
			$\text{NIG}_1$ & -0.0029 &  0.0030 & 0.028 && -0.0030 &  0.0031 & 0.026 &\\ 
			$\text{NIG}_2$ & -0.0035 &  0.0037 & 0.026 && -0.0037 &  0.0039 & 0.025 &\\ 
			$\text{NIG}_3$ & -0.0013 &  0.0017 & 0.029 && -0.0017 &  0.0023 & 0.024 &\\ 
			$\text{NIG}_4$ & -1.0283 &  1.4994 & 0.074 && -1.8095 &  2.3915 & 0.068 &\\ 
			\bottomrule
		\end{tabular} \\
		\caption{Polynomial risk measures with $\gamma=4$ and $\gamma=5$. Time in seconds.}
		\label{Table_poly_4-5}
	\end{center}
\end{table} 

We can immediately observe that the combination of a deterministic root-finding
algorithm with the Fourier representation for the optimal allocation and risk
measure yields numerical results in very short time for all combinations of
parameters and exponents. In general, less than 1/10 of a second is required to
solve the optimization problem corresponding to the allocation $\eta^*$ and then
to compute the risk measure. One can also observe that computational times are
decreasing as the relative risk aversion parameter $\gamma$ increases.

In order to compare our results, we have used a stochastic root finding (SRF)
algorithm, see
\citep{Bardou_Frikha_Pages_2009,Dunkel_Weber_2010,Hamm_Salfeld_Weber_2012}. We
use 30,000 iteration steps as suggested by the results in
\cite{Hamm_Salfeld_Weber_2012}, although we have not implemented a variance
reduction technique. Note that, for a fixed number of steps, implementation of 
a variance reduction technique would increase the computational time. The 
computational times for the stochastic root finding methods in all datasets for 
$\gamma=2$ are reported in the last column of Table \ref{Table_poly_2}. The 
times for the other exponents are almost identical, thus have been omitted for 
the sake of brevity. One can immediately observe that the combination of 
deterministic root finding methods with Fourier representation is several times 
faster than the stochastic root finding schemes. In the worst case, the factor 
equals 4, while in most cases it exceeds 10. Apart from the gains in 
computational time, it should be stressed that the Fourier method yields an 
exact value for both $\eta^*$ and $\rho(X)$, while the stochastic root finding 
scheme delivers only an estimate. 

\bibliographystyle{abbrvnat}
\bibliography{bibliography}
\end{document}